\documentclass[11pt,a4paper,draft]{article}
\usepackage{geometry}
\usepackage[latin1]{inputenc}

\usepackage{amsmath,amsthm,amssymb,subdepth}
\usepackage{amsfonts,mathrsfs,enumerate}
\usepackage{setspace}
\usepackage{hyperref}

\title{First Law of Black Hole Mechanics as a Condition for Stationarity}
\date{\today}
\author{Stephen\hspace{-1.5mm}\newlength{\Mheight}
\newlength{\cwidth}
\settoheight{\Mheight}{M}\settowidth{\cwidth}{c}M\parbox[b][\Mheight][t]{\cwidth}{c}\hspace{0.15mm}Cormick\footnote{stephen.mccormick@une.edu.au}\vspace{3mm}\\School of Science and Technology\\University of New England\\Armidale, 2351\\Australia}
\newtheorem{theorem}{Theorem}[section]

\newtheorem{proposition}[theorem]{Proposition}

\newtheorem{corollary}[theorem]{Corollary}

\newtheorem{lemma}[theorem]{Lemma}

\newtheorem{prop}[theorem]{Proposition}
\newtheorem{remark}[theorem]{Remark}

\newcommand{\onabla}{\mathring\nabla}
\newcommand{\og}{\mathring g}
\newcommand{\lie}[1]{\mathfrak{#1}}
\newcommand{\Jxi}{\tilde{J}_{\xi_{\tref}}}

\newcommand{\xiref}{\xi_{\tref}}
\newcommand{\st}{\hspace{1mm}|\hspace{1mm}}

\newcommand{\fd}{{}^4\hspace{-0.6mm}}
\newcommand{\gpi}{g,A,\pi,\varepsilon}
\newcommand{\Mo}{{\mathcal{M}_0}}
\newcommand{\hatH}{\hat{\mathcal{H}}}
\newcommand{\Pbundle}{\Lambda^0(\mathcal{M})\times T\mathcal{M}\times\lie{g}\otimes\Lambda^0(\mathcal{M})}

\numberwithin{equation}{section}

\DeclareMathOperator{\tr}{tr}

\DeclareMathOperator{\tref}{ref}

\DeclareMathOperator{\Area}{Ar}
\usepackage{fancyhdr}
\begin{document}
\maketitle
\begin{abstract}
In earlier work, we provided a Hilbert manifold structure for the phase space for the Einstein-Yang-Mills equations, and used this to prove a condition for initial data to be stationary \cite{Me}. Here we use the same phase space to consider the evolution of initial data exterior to some closed $2$-surface boundary, and establish a condition for stationarity in this case. It is shown that the differential relationship given in the first law of black hole mechanics is exactly the condition required for the initial data to be stationary; this was first argued non-rigorously by Sudarsky and Wald in 1992 \cite{SW1}. Furthermore, we give evidence to suggest that if this differential relationship holds then the boundary surface is the bifurcation surface of a bifurcate Killing horizon.

\end{abstract}

\section{Introduction}
In 1992, Sudarsky and Wald \cite {SW1} discussed the first law of black hole mechanics in the context of Einstein-Yang-Mills theory. Among other things, they noted that certain surface integrals, associated with the Hamiltonian, were closely related to the first law. From this, it was argued that the differential relationship given by the first law provides a condition for stationarity of the Einstein-Yang-Mills equations. This argument was based on earlier work by Brill, Deser and Fadeev \cite{BDF}, who proposed in the pure Einstein case, that stationary solutions were exactly those solutions that extremise the ADM mass over the space of solutions. Both arguments were based on Lagrange multipliers, however neither provided the mathematical machinery required to make such an argument rigorous. The essential missing ingredient, to develop this argument into a mathematical proof, is a manifold structure for the space of solutions.

In 2005, Bartnik \cite{phasespace} provided such a Hilbert manifold structure for the Einstein case, and from this a complete proof of the Brill, Deser and Fadeev argument was given. At first, this may appear to contradict the argument of Sudarsky and Wald, since we have that a solution is stationary if and only if it is a critical point of the mass. However, the case considered by Bartnik has no Maxwell or Yang-Mills fields, and the initial data manifold has a single asymptotic end with no interior boundary; in this case, the first law simply reduces to $dm=0$. Recently, using similar ideas, the Einstein-Yang-Mills case has been considered by the author \cite{Me}. As this article builds on many things in \cite{Me}, we will refer to this throughout as Paper I. In Paper I, a suitable phase space for the Einstein-Yang-Mills equations is outlined, and the condition for stationarity in this case becomes
\begin{equation}
dm+V_\infty\cdot dQ_\infty=0,\label{firstlawa}
\end{equation}
where $V_\infty$ is the electric potential at infinity and $Q_\infty$ is the total electric charge. Again, (\ref{firstlawa}) is the appropriate first law in this case.

In this article, we consider evolution exterior to some closed $2$-surface boundary, and conclude that the condition for stationarity is again the appropriate version of the first law; namely,
\begin{equation}
d m=\frac{\kappa}{8\pi}dA + \Omega d J + V\cdot d Q-V_\infty\cdot d Q_\infty,
\end{equation}
where $A$ is the area, $\Omega$ is the angular velocity, $J$ is the angular momentum, $V$ is the electric potential and $Q$ is the electric charge, of the boundary surface respectively. Note that the term, $V_\infty\cdot d Q_\infty$, not generally included in the first law, permits non-zero electric potential at infinity. In the Maxwell electrovacuum case, $Q=Q_\infty$, so the expression $(V\cdot d Q-V_\infty\cdot d Q_\infty)$ is equivalent to $\tilde{V}dQ$, where $\tilde{V}=V-V_\infty$ is the potential difference between the boundary surface and infinity.

An initial data set for the Einstein-Yang-Mills equations is a tuple $(g,A,\pi,\varepsilon)$; a Riemannian metric, a Lie algebra-valued one form, a symmetric covariant 2-tensor density and a Lie coalgebra-valued vector density on a 3-manifold, $\mathcal{M}$. Here $\pi$ is the usual momentum conjugate to $g$, $A$ is the gauge field projected onto $\mathcal{M}$ and $\varepsilon$ is its associated momentum, equal to negative four times the Yang-Mills electric field density, $E$. We also make use of the quantity 
\begin{equation*}
B^i_a:=\frac{1}{2}\epsilon^{ijk}(\nabla_j A_{ak}-\nabla_k A_{aj} +C_{abc}A^b_j A^c_k),\label{Bdefn}
\end{equation*}
the Yang-Mills magnetic field density, where $\epsilon^{ijk}$ is the usual antisymmetric tensor density and $C_{abc}$ are the structure constants of the Lie algebra, and the indices $a,b,c...$ are Lie algebra indices (see (\ref{indextable})). Note that the connection in (\ref{Bdefn}) can be replaced by any torsion-free connection, due to the antisymmetry in $\nabla A$. Throughout, we consider the electric and magnetic fields as viewed by a Gaussian normal set of observers; that is, observers whose worldlines are orthonormal to the Cauchy surface.

The Hilbert manifold structure from Paper I, considered here, consists of initial data sets $(g,A,\pi,\varepsilon)$ with local regularity $H^2\times H^2\times H^1 \times H^1$ and appropriate decay for an asymptotic flatness. It is interesting to note that this is exactly the regularity required by recent work of Klainerman, Rodnianski and Szeftel \cite {boundedl2} to ensure that the Cauchy problem for the Einstein equations is well-posed. Furthermore, the regularity assumptions on the Yang-Mills initial data is exactly that required to ensure that the Cauchy for the Yang-Mills equations on a curved background is well-posed, which was recently demonstrated by Ghanem \cite{ghanem}. To the best of the author's knowledge the Cauchy problem for the coupled system has not been considered at this regularity, however given that each system is well-posed independently, one expects the coupled system to also be well-posed.

The outline of this article is as follows. In Section \ref{Ssetup}, we recall the phase space and constraint submanifold from Paper I. Section \ref{Smass} introduces the mass, charge and angular momentum definitions, and establishes some properties of these quantities as functions on the phase space. Finally, in Section \ref{Shamiltonian}, we discuss Hamiltonians and use a Lagrange multiplier argument to establish the condition for stationarity.

%%%%      SECTION STARTS HERE       %%%%%%%%%%%%%%%%%%%%%%%%%%%%%%%%%%%%%
\section{The Phase Space}\label{Ssetup}
Let $\mathcal{M}$ be a complete, paracompact, connected, oriented 3-manifold, which is asymptotically flat in the following sense: There exists a compact set, $K\subset\mathcal{M}$, such that ${\mathcal{M}\setminus K=\bigcup_{i=1}^N M_i}$, with each $M_i$ diffeomorphic to $\mathbb{R}^3$ minus the closed unit ball; there exists a collection of diffeomorphisms $\phi_i:M_i\rightarrow \mathbb{R}^3\setminus \overline{B_1(0)}$. On $\mathcal{M}$, fix a smooth background metric $\og$ such that $\og\equiv\phi_{i*}(g_{\mathbb{R}^3})$ on each $M_i$, the pullback of the Euclidean metric. Further, define a smooth function $r(x)\geq1$ on $\mathcal{M}$, such that $r(x)=|\phi_i(x)|$ on each $M_i$.

Next, recall the weighted Lebesgue and Sobolev spaces, which describe the phase space. The spaces $L^p_\delta(\mathcal{M})$ and $W^{k,p}_\delta(\mathcal{M})$ are defined respectively as the completion of $C^\infty_c(\mathcal{M})$ with respect to the norms
\begin{equation}
\|u\|_{p,\delta}:=\left(\int_\mathcal{M}|u|^p r^{-\delta p -n}\mathring{d\mu}\right)^{1/p},\hspace{9mm}\left\|u\right\|_{k,p,\delta}:=\sum_{j=0}^k\|\onabla^j u\|_{p,\delta-j}.
\end{equation}
We use $\circ$ to denote quantities determined by $\og$, such as the background Levi-Civita connection, $\onabla$, and measure, $\mathring{d\mu}=\sqrt{\og}dx^3$. Weighted Lebesgue and Sobolev spaces of sections of bundles are defined in the usual way. These weighted spaces have the same local regularity as the usual Lebesgue and Sobolev spaces and behave as $o(r^{\delta})$ near infinity on each of the ends, with each each successive derivative decaying one power of $r$ faster. Refer to \cite{AF, ellipticsys, NW} for details on the weighted spaces.

The Yang-Mills gauge group is taken to be a compact Lie group, $G$, with Lie algebra, $\lie{g}$. We identify $\lie{g}$ with its Lie coalgebra, $\lie{g}^*$, via a positive definite inner product, $\gamma$, which may be taken to be the negative of the Killing form on the semisimple factor and the usual Euclidean inner product on the abelian factor. The usual decay conditions for asymptotic flatness and the regularity assumptions mentioned above suggest that we impose $(g-\og)\in W^{2,2}_{-1/2}$ and $\pi\in W^{1,2}_{-3/2}$, noting that $\pi$ behaves like a derivative of the metric. Imposing $\varepsilon\in W^{1,2}_{-3/2}$ enforces the usual $\frac{1}{r^2}$ fall off of the electric field in electromagnetism, however the appropriate domain for $A$ is less obvious. The Lie algebra, $\lie{g}$, is split into its centre, $\lie{z}$, and a $\gamma$-orthogonal subspace, $\lie{k}$. Then $A$ is decomposed into $A=A_\lie{z}+A_\lie{k}$, with $A_\lie{z}$ valued in $\lie{z}$ and $A_\lie{k}$ valued in $\lie{k}$. The domain for $A$ is taken to be such that $A_\lie{z}\in W^{2,2}_{-1/2}$ and $A_\lie{k} \in W^{2,2}_{-3/2}$.

The decay conditions on $A$ are chosen such that the gauge covariant derivative, ${\hat{D}:=\partial + [A,\cdot]\sim\partial +A_{\lie{k}}}$, behaves analogously to the usual covariant derivative at infinity; that is, $\hat{D}\theta=\partial\theta+o(r^{-3/2})\theta$. Although it may appear somewhat unnatural to require this condition for the analysis, such a condition is in fact required to ensure that the total charge is well-defined \cite{globalcharges}. It should be noted that this condition also puts the electric and magnetic fields on equal footing. In the language of physics, this condition is that the Yang-Mills fields are asymptotic to photon fields before vanishing.

Formally, the phase space from Paper I is given by
\begin{equation*}
\mathcal{F}:=\mathcal{G}^+\times\mathcal{K}\times\mathcal{A}\times\mathcal{E},
\end{equation*}
where
\begin{align*}
\mathcal{G}^+:&=\{g\hspace{1mm}|\hspace{1mm}(g-\og)\in W^{2,2}_{-1/2}(S_2),\hspace{2mm}g>0\},&
\mathcal{K}:&=W^{1,2}_{-3/2}(S^2\otimes\Lambda^3),\\
\mathcal{A}:&=W^{2,2}_{-1/2}(T^*\mathcal{M}\otimes \lie{z})\oplus W^{2,2}_{-3/2}(T^*\mathcal{M}\otimes \lie{k}),&
\mathcal{E}:&=W^{1,2}_{-3/2}(T\mathcal{M}\otimes \lie{g}^*\otimes\Lambda^3).
\end{align*}
In the above, $S_2$ and $S^2$ are the spaces of symmetric covariant and contravariant tensors on $\mathcal{M}$ respectively, and we denote by $\Lambda^k$, the bundle of $k$-forms on $\mathcal{M}$.

We also define the spaces
\begin{align*}
\mathcal{N}:&=L^2_{-1/2}(\Lambda^0\times T\mathcal{M}\times\mathfrak{g}\otimes\Lambda^0),\\
\mathcal{N}^*:&=L^2_{-5/2}(\Lambda^3\times T^*\mathcal{M}\otimes\Lambda^3\times \mathfrak{g}^*\otimes\Lambda^3).
\end{align*}
Throughout this article, we use the following conventions for indices on different spaces:
\begin{align}\label{indextable}
\begin{tabular}{|r||c    l|}
\hline
  $\mathcal{M}$, $\mathbb{R}^3$&Latin lower case, mid-alphabet &$i,j,...$\\\hline
  ${}^4\hspace{-0.6mm}\mathcal{M}$, $\mathbb{R}^{3,1}$ &Greek lower case, mid-alphabet  & $\mu,\nu...$  \\ \hline
  $\lie{g}$ &Latin lower case, early alphabet &$a,b...$    \\  \hline
  ${}^4\hspace{-0.6mm} P$, $(\mathbb{R}^{3,1}\oplus\lie{g})$ &Greek lower case, early alphabet &$\alpha,\beta...$     \\  \hline
\end{tabular},
\end{align}
where ${}^4\hspace{-0.6mm} \mathcal{M}$ is the spacetime in which $\mathcal{M}$ sits, and ${}^4\hspace{-0.6mm} P$ is a $G$-bundle over ${}^4\hspace{-0.6mm} \mathcal{M}$, which is associated with the Yang-Mills fields. By a slight abuse of notation, we will write ${\xi^\alpha=(\xi^0,\xi^i,\xi^a)=(\xi^\mu,\xi^a)}$ to indicate a $(4+n)$ dimensional object, and identify the components with appropriate projections. For example, if $\xi^\alpha$ is a section of $T{}^4\hspace{-0.6mm} P$, we consider $\xi^0$ to be a scalar function, $\xi^i$ to be a vector field over $\mathcal{M}$, and $\xi^a\in\lie{g}$.

Recall the constraint map, $\Phi:\mathcal{F}\rightarrow \mathcal{N}^*$, given by
\begin{alignat}{3}
\Phi_0(g,A,\pi,\varepsilon)&=(\frac{1}{2}(\pi^k_k)^2-\pi^{ij}\pi_{ij}-(\frac{1}{8}\varepsilon^k_a \varepsilon_k^a+2B^k_a B^a_k))g^{-1/2}+R\sqrt{g} \label{constraint1}, \\
\Phi_i(g,A,\pi,\varepsilon)&=2\nabla^j\pi_{ij}-\varepsilon^j_a(\onabla_i A^a_j-\onabla_j A^a_i)+\onabla_j(\varepsilon^j_a)A_i^a\label{constraint2},\\
\Phi_a(g,A,\pi,\varepsilon)&=-\onabla_j\varepsilon^j_a-C^c_{ab}A^b_j\varepsilon^j_c \label{gauss}.
\end{alignat}
The momentum constraint (\ref{constraint2}) differs from that considered in Paper I by the term $\Phi_a A_i^a$. This difference amounts to a difference in interpretation of the non-dynamical degree of freedom associated with $\Phi_a$. As this is simply the addition of another constraint, the results of Paper I clearly remain valid. Also note that in Paper I, $\mathcal{M}$ was considered to have only a single asymptotic end, however this was for simplicity of presentation rather than technical reasons. It is clear that the entire phase space analysis is valid for multiple asymptotic ends (the full analysis is presented in Chaper 4 of the author's doctoral thesis \cite{MyThesis}). In particular, for a given source, $s\in\mathcal{N}^*$, the level set
\begin{equation*}
\mathcal{C}(s):=\{(g,A,\pi,\varepsilon)\in\mathcal{F}\st\Phi(g,A,\pi,\varepsilon)=s\},
\end{equation*}
has a Hilbert manifold structure; we call this the constraint submanifold. We demonstrate that the energy-momentum and other quantities are not defined on all of $\mathcal{F}$, so we will view the energy, momentum, angular momentum and charge as functions on constraint submanifolds with integrable sources.

%%%%% NEW SECTION %%%%%%%%%%%%%%%%%%%%%%%%%%%%%%
\section{Mass, Charge and Angular Momentum}\label{Smass}
In this Section, we discuss the quantities relevant to the first law; some of which are defined at a particular end, and others on some surface to later correspond to a horizon. In order to do this, an artificial boundary to one of the ends is introduced. Let $\Sigma$ be a closed 2-surface such that $\mathcal{M}\setminus\Sigma$ consists of two connected components; one of which contains only a single end, $M_0$. Denote by $\mathcal{M}_0$, the connected component of $\mathcal{M}\setminus\Sigma$ that contains $M_0$.

The ADM energy-momentum covector, ${\mathbb{P}_\mu(g,\pi)=(\mathbb{P}_0,\mathbb{P}_i)=(m_0,p_i)}$, is given by
\begin{align}
16\pi m_0&:=\oint_{S_\infty}\og^{jk}(\onabla_kg_{ij}-\onabla_i g_{jk})dS^i,\\
16\pi p_i&:=2\oint_{S_\infty}\pi_{ij}dS^j,
\end{align}
where $S_\infty$ is understood as the limit of increasingly large spheres. Throughout, the unit normal vector associated with the surface element $dS$ is to be understood as pointing in the direction of infinity in $M_0$. The $\lie{g}$-valued total Yang-Mills electric charge is given by
\begin{equation}
16\pi Q_{\infty\, a}:=4\oint_{S_\infty}E_a^idS_i=-\oint_{S_\infty}\varepsilon_a^idS_i,
\end{equation}
and we write $\mathbb{P}_a=Q_{\infty\, a}$, so that the tuple $\mathbb{P}_\alpha:=(\mathbb{P}_0,\mathbb{P}_i,\mathbb{P}_a)\in\mathbb{R}^{3,1}\oplus\lie{g}^*$ can be identified with the asymptotic value of a section of ${}^4\hspace{-0.6mm} P$. The charge, $Q_\Sigma$, associated with $\Sigma$, is defined analogously,
\begin{equation}
16\pi Q_{\Sigma\, a}:=4\oint_{\Sigma}E_a^idS_i=-\oint_{\Sigma}\varepsilon_a^idS_i.
\end{equation}

Let $\xi^\mu_\infty\in\mathbb{R}^{3+1}$ be identified with some timelike vector, corresponding to the tangent to the worldline of an observer at spatial infinity. Further, let $\xi^a_\infty\in\lie{g}$ correspond to the asymptotic value of the electric potential, which we will assume to be constant. A total measure of the energy, viewed by this observer, is then given by $\xi_\infty\cdot (E,p_i,Q_a)$, which will be more convenient to work with than the tuple, $(E,p_i,Q_a)$, itself. In order to write this as the integral of a divergence, we need to make sense of extending $\xi_\infty$ to a section of $T{}^4\hspace{-0.6mm} P\cong\Pbundle$.

Near infinity, $\xi_\infty\in \mathbb{R}^{3,1}\oplus\lie{g}$ may be identified with some smooth section, 
\begin{equation*}
\tilde{\xi}_\infty\in C^\infty(\Pbundle),
\end{equation*}
such that $\onabla \tilde{\xi}_\infty=0$. We then say a smooth section, $\hat{\xi}_\infty\in C^\infty(\Pbundle)$, is a constant translation near infinity representing $\xi_\infty$, if ${\hat{\xi}_\infty=\tilde{\xi}_\infty}$ on $E_{2\hat{R}}$ and vanishes on $B_{\hat{R}}$, for some $\hat{R}$, where $B_R:=\{x\in\mathcal{M}\st r(x)<R\}$ and $E_R:=\mathcal{M}_0\setminus \overline{B_R}$. While a representation of $\xi_\infty$ is not unique, the difference between two distinct representations is smooth and compactly supported. This lets us prescribe asymptotics for $\xi$, but we would also like to prescribe some boundary values on $\Sigma$; for this, fix a smooth section, $\hat{\xi}_\Sigma$, with support near $\Sigma$. We then define $\xi_{\tref}:=\hat{\xi}_\infty+\hat{\xi}_\Sigma$ to encapsulate both boundary conditions.

Define the spaces
\begin{align}
W^{2,2}_{\xi_{\tref}}:&=\big{\{} \xi \st (\xi-\xi_{\tref}) \in W^{2,2}_{-1/2 \, c}({\Lambda^0(\mathcal{M}_0)\times T\mathcal{M}_0\times\lie{g}\otimes\Lambda^0(\mathcal{M}_0)})\big{\}},\\
L^{2}_{\xi_\infty}:&=\big{\{} \xi \st (\xi-\hat{\xi}_\infty) \in L^{2}_{-1/2}(\Lambda^0(\mathcal{M}_0)\times T\mathcal{M}_0\times\lie{g}\otimes\Lambda^0(\mathcal{M}_0))\big{\}},
\end{align}
where $W^{2,2}_{-1/2 \, c}$ is the completion of $C^\infty_c$ with respect to the $W^{2,2}_{-1/2}$ norm. Elements of these spaces may be interpreted as sections of $\fd P$, restricted to $\mathcal{M}_0$, with prescribed asymptotics and boundary values on $\Sigma$.

Setting $\hat{\xi}^0_\Sigma\equiv 0$, we define the energy-momentum covector by its pairing with with a vector at infinity, as follows:
\begin{align}
16\pi\xi^0_\infty\mathbb{P}_0(g)=&\int_{\mathcal{M}_0}\Big{(}\hat{\xi}^0_\infty\og^{ik}\og^{jl}(\onabla_k\onabla_l g_{ij}-\onabla_i\onabla_k g_{jl})\\
&+\og^{ik}\og^{jl}\onabla_k\hat{\xi}^0_\infty(\onabla_l g_{ij}-\onabla_i g_{jl})\Big{)}\sqrt{\mathring{g}}\label{ADME},\\
16\pi\xi^i_\infty\mathbb{P}_i(\pi)=&\int_{\mathcal{M}_0}\left(2\xi^i_{\tref}\onabla_j\pi_i^j+2\pi^{ij}\onabla_i\xi_{\tref j}+\onabla_i(\varepsilon^i_a A^a_j) \xi_{\tref}^j+\varepsilon^i_a A^a_j\onabla_i \xi_{\tref}^j\right)\nonumber\\
&+\oint_\Sigma \left(2\xi_\Sigma^i\pi^j_i -\varepsilon^j_a A^a_i \xi_{\Sigma}^i\right)dS_j,\label{JPterms}
\end{align}
Note that while (\ref{JPterms}) contains the terms $(g,A,\varepsilon)$, the quantity $\mathbb{P}_i$ only depends only on $\pi$; the boundary terms on $\Sigma$ combine with the bulk integral to give a boundary integral at infinity, which removes the dependence on $g$ as $g=\og+o(r^{-1/2})$, and the Yang-Mills terms at infinity vanish (see (\ref{vanishYM})), leaving only $\pi$ dependence. When $\hat{\xi}^i_\Sigma$ agrees with a rotational Killing field, the integral over $\Sigma$ in (\ref{JPterms}) is proportional to the angular momentum. This leads us to define a generalised notion of angular momentum,
\begin{equation}
16\pi\tilde{J}_{\xi_{\tref}}(g,A,\pi,\varepsilon):=-\oint_\Sigma \left(2\hat{\xi}_\Sigma^i\pi^j_i -\varepsilon^j_a A^a_i \hat{\xi}_{\Sigma}^i\right)dS_j.\label{Jdef}
\end{equation}

Note that we follow the sign convention of Wald \cite{Wald}. The second term in ($\ref{Jdef}$), corresponding to the angular momentum of the Yang-Mills fields, is non-standard and appears to have been first considered by Sudarsky and Wald \cite{SW1}, however they considered the integration to be performed at infinity. It will be important for us to use a quasilocal\footnote{While this is useful for our purposes, we do not argue here that this gives a suitable quasilocal definition of angular momentum in general. There is a great deal of literature on the problem of quasilocal mass and angular momentum (see \cite{QLMreview} and references therein).} definition of angular momentum instead.

To write the electric charge as a bulk integral, we will fix a choice of the Lagrange multiplier, $\xiref^a$, with $\hat{\xi}_\Sigma=\xi_\Sigma\in\lie{g}$, constant. Similar to the above, we have
\begin{equation}
16\pi(\xi_\infty^a \mathbb{P}_a-\xi^a_\Sigma Q_{\Sigma\, a})=4\int_{\mathcal{M}_0}\left(\xiref^a\onabla_i E^i_a +E^i_a\onabla_i\xiref^a\right)\label{Qterms}.
\end{equation}

\begin{lemma}\label{LemQJ}
Let $\chi$ be a vector field on $\mathcal{M}$ with $\|\chi\|_{L^\infty(\Sigma)}<\infty$. The maps $Q_\Sigma:\mathcal{F}\rightarrow\lie{g}^*$ and $\tilde{J}_{\chi}:\mathcal{F}\rightarrow\mathbb{R}$ are smooth.
\end{lemma}
\begin{proof}
By considering any function $\varphi\in C^\infty_c(\mathcal{M})$ with $\varphi\equiv 1$ on $\Sigma$, the Sobolev trace theorem gives
\begin{equation}
|Q_\Sigma|\leq c\|E\|_{L^1(\Sigma)}=\|\varphi E\|_{L^1(\Sigma)}\leq c\|\varphi E\|_{L^2(\Sigma)}\leq c\|E\|_{1,2,-3/2}.
\end{equation}
We estimate $\tilde{J}_{\chi}$ similarly:
\begin{align*}
\tilde{J}_{\chi}&\leq c(\|\chi\|_{L^2(\Sigma)}\|\pi\|_{1,2,-3/2}+\|\chi\|_{L^\infty(\Sigma)}\|\varphi A\|_{L^2(\Sigma)}\|\varphi \varepsilon\|_{L^2(\Sigma)})\\
&\leq c\|\chi\|_{L^\infty(\Sigma)}(\|\pi\|_{1,2,-3/2}+\|A\|_{1,2,-1/2}\|\varepsilon\|_{1,2,-3/2}).
\end{align*}
Since $Q_\Sigma$ and $\tilde{J}_{\chi}$ are bounded and linear, smoothness follows. 
\end{proof}
\begin{theorem}\label{Psmooth}
For an integrable source, $s\in L^1$, the map $\mathbb{P}:\mathcal{C}(s)\rightarrow\mathbb{R}^{3.1}\oplus\lie{g}^*$ is smooth.
\end{theorem}
\begin{proof}
$\mathbb{P}_0$ is exactly of the form considered by Bartnik $\cite{phasespace}$, except that the integrals are over a manifold with boundary in our case. However, this difference does not affect Bartnik's proof that $\mathbb{P}_0$ is smooth so the result applies here also. $\mathbb{P}_i$ differs from Bartnik's by some Yang-Mills terms and the term $16\pi\Jxi$, so we consider it again here. Lemma \ref{LemQJ} shows that $\Jxi$ is smooth, so we must only consider the bulk (volume) integral, which is shown to be smooth by the same reasoning as that used by Bartnik. Note that the second and fourth terms in the bulk integral defining $\mathbb{P}_i$ (\ref{JPterms}) are clearly bounded as $\onabla\xi$ has bounded support. The remaining two terms are estimated as follows:
\begin{equation*}
\int_{\mathcal{M}_0}2\xi^i_{\tref}\onabla_j\pi_i^j\leq c\|\xi^i_{\tref}\|_{\infty,0}\|\onabla\cdot\pi\|_{1,-3},
\end{equation*}
which is then controlled by the fact that we have an integrable source. Recalling the difference of connections tensor,
\begin{equation*}
\tilde{\Gamma}^i_{jk}:=\Gamma^i_{jk}-\mathring{\Gamma}^i_{jk}=\frac{1}{2}g^{il}(\onabla_j g_{lk}+\onabla_k g_{jl}-\onabla_l g_{jk}),
\end{equation*}
and making use of the momentum constraint (\ref{constraint2}), we have
\begin{align*}
\|\onabla\cdot\pi\|_{1,-3}\leq&\, c(\|\nabla\cdot\pi\|_{1,-3}+\|\tilde{\Gamma}\pi\|_{1,-3})\\
\leq &\, c(\|s\|_{1,-3}+\|\varepsilon\onabla A\|_{1,-3}+\|A\onabla\varepsilon\|_{1,-3}+\|\tilde{\Gamma}\|_{2,-3/2}\|\pi\|_{2,-3/2})\\
\leq &\, c\big{(}\|s\|_{1,-3}+\|\varepsilon\|_{2,-3/2}\|\onabla A\|_{2,-3/2}\\
&+\|A\|_{2,-1/2}\|\onabla\varepsilon\|_{2,-5/2}+\|\onabla g\|_{2,-3/2}\|\pi\|_{2,-3/2}\big{)}\\
\leq &\, c(\|s\|_{1,-3}+\|\varepsilon\|_{1,2,-3/2}\|A\|_{1,2,-1/2}+\|\onabla g\|_{2,-3/2}\|\pi\|_{2,-3/2}).
\end{align*}
Similarly, we have
\begin{align}
\int_{\mathcal{M}_0}\xiref^j\onabla_i(\varepsilon^i_aA^a_j)&\leq c\|\xiref\|_{\infty,0}(\|A\onabla\varepsilon\|_{1,-3}+\|\varepsilon\onabla A\|_{1,-3})\nonumber\\
&\leq c\|\xiref\|_{\infty,0}(\|A\|_{2,-1/2}\|\onabla\varepsilon\|_{2,-5/2}+\|\varepsilon\|_{2,-3/2}\|\onabla A\|_{2,-3/2})\label{momentumterminham}.
\end{align}
Since the bulk integral is linear in each of the variables and bounded, smoothness follows; that is, $\mathbb{P}_i$ is smooth.

The remaining component, $\xi^a_\infty\mathbb{P}_a$, consists of a bulk integral plus the term $\xi_\Sigma^aQ_{\Sigma\, a}$ (\ref{Qterms}); the latter is again smooth by Lemma \ref{LemQJ} and the bulk integral is estimated similarly to the above. The second term in the bulk integral is clearly bound again as $\onabla\xiref$ has bounded support, and the first term makes use of the Gauss constraint (\ref{gauss}) and the fact that the source is integrable,
\begin{align*}
\int_{\mathcal{M}_0}\xiref^a\onabla_iE^i_a&\leq c(\|\xiref\|_{\infty,0}\|\onabla\cdot E\|_{1,-3})\\
&\leq c\|\xiref\|_{\infty,0}(\|s\|_{1,-3}+\|A_{\lie{k}}\varepsilon\|_{1,-3})\\
&\leq c\|\xiref\|_{\infty,0}(\|s\|_{1,-3}+\|A_{\lie{k}}\|_{2,-3/2}\|\varepsilon\|_{2,-3/2}).
\end{align*}
It follows that $\mathbb{P}$ is smooth.
\end{proof}

%NEW SECTION STARTS HERE%%%%%%%%%%%%%%%%%%%%%%

\section{Hamiltonians and The First Law}\label{Shamiltonian}
It is well-known that the source-free evolution equations can be succinctly written as
\begin{equation}
\frac{d}{dt}\begin{bmatrix}g\\A\\ \pi\\ \varepsilon\end{bmatrix}=-\begin{bmatrix}0&0&1&0\\0&0&0&1\\-1&0&0&0\\0&-1&0&0\end{bmatrix} \circ D\Phi_{(g,A,\pi,\varepsilon)}^*(\xi),\label{hamiltonseq}
\end{equation}
where $D\Phi_{(g,A,\pi,\varepsilon)}^*$ is the formal adjoint of the linearisation of $\Phi$, and $t$ is interpreted as the flow parameter of a vector field on ${}^4\hspace{-0.6mm} P$, identified with $\xi$ (see, for example, \cite{ArmsEYM,etal1}). The flow of $\xi$ is interpreted as a simultaneous time-evolution and continuous change of gauge. Equation (\ref{hamiltonseq}) motivates Moncrief's result, equating stationary solutions, with initial data satisfying $D\Phi_{(g,A,\pi,\varepsilon)}^*(\xi)=0$ for some $\xi^\mu$ corresponding to a time translation at infinity \cite{Moncrief1,Moncrief2} (see also the subsequent work by Arms, Marsden and Moncrief in the Einstein-Yang-Mills case \cite{etal1}). Such an initial data set, we call a \textit{generalised stationary initial data set}.

If the formal adjoint agrees with the true adjoint, then these evolution equations correspond exactly to Hamilton's equations for the usual ADM Hamiltonian,
\begin{equation}
\mathcal{H}^{ADM\,(\xi)}(g,A,\pi,\varepsilon):=-\int_\mathcal{M}\xi\cdot \Phi(g,A,\pi,\varepsilon)\label{HADM}.
\end{equation}
Unfortunately, this is not the case when $\mathcal{M}$ is an asymptotically flat manifold as the formal adjoint differs from the true adjoint by a collection of boundary terms unless $\xi$ vanishes sufficiently fast at infinity. In order to generate the correct equations of motion, the first variation of the Hamiltonian density must be of the form
\begin{equation}
D H_{(g,A,\pi,\varepsilon)}(h,b,p,f)=\Xi\cdot(h,b,p,f)\label{hamform},
\end{equation}
for some $\Xi \in T^*_{(g,A,\pi,\varepsilon)}\mathcal{F}$.

In the pure gravity case with no interior boundary, Regge and Teitelboim demonstrated that by adding the ADM mass to the ADM Hamiltonian, the correct equations of motion are obtained \cite{RT}. In Paper I, where we consider the Einstein-Yang-Mills case with no interior boundary, we also add a charge term, corresponding to the additional Yang-Mills energy. However, problems arise when one looks at the evolution exterior to some boundary.

Na\"ively using the ADM Hamiltonian density (\ref{HADM}), we find
\begin{align}
D H^{ADM\,(\xi)}_{(g,A,\pi,\varepsilon)}(h,b,p,f)=&-D\Phi^*_{(g,A,\pi,\varepsilon)}(\xi)\cdot(h,b,p,f)+\nabla^i\Big{(}(\xi^0(\onabla_i\text{tr}_gh-\nabla^jh_{ij})\nonumber \\
&+\onabla^j(\xi^0)h_{ij}-\text{tr}_g h\onabla_i(\xi^0))\sqrt{g}-2\xi^j p_{ij}+\xi^a f_{ai}-2\pi^k_i h_{jk}\xi^j\nonumber\\
&+ \pi^{jk}h_{jk}\xi_i -\epsilon_{ijk}b^{ak}B^j_a \xi^0\sqrt{g}-\varepsilon_{ia}b^a_j\xi^j+\xi_i\varepsilon^j_ab^a_j-f_{ia}A_j^a \xi^j\Big{)}\label{DHADM}.
\end{align}
The first term is exactly of the form we require (\ref{hamform}), however the cumbersome divergence term does not vanish in general. Fortunately, it does have the following geometric interpretation to be exploited. Let $\Sigma$ be the bifurcation surface of a bifurcate Killing horizon, $\xi^\mu$ be the stationary Killing field and $\phi^\mu$ be the rotational Killing field tangent to $\mathcal{M}$ with $2\pi$-periodic orbits; we then have $\xi^\mu+\Omega\phi^\mu\equiv 0$ on $\Sigma$ for some constant $\Omega$, which is to be interpreted as angular velocity of the horizon. The zeroth law of black hole mechanics states that the surface gravity $\kappa=\frac{1}{2}n^i\nabla_i\xi^0$ is constant on $\Sigma$, where $n^i$ is the unit normal to $\Sigma$ pointing towards infinity in $M_0$. We also ask that the electric potential, $V^a=\xi^a$ be constant at infinity and on $\Sigma$. In this case, the expression (\ref{DHADM}) becomes
\begin{align}
\int_{\mathcal{M}_0}D H^{ADM\,(\xi)}_{(g,A,\pi,\varepsilon)}(h,b,p,f)=&-\int_{\mathcal{M}_0}D\Phi^*_{(g,A,\pi,\varepsilon)}(\xi)\cdot(h,b,p,f)-16\pi Dm_{(g,\pi)}(h,p)\nonumber\\
&+2\kappa D\Area_{\Sigma\, g}(h)
+16\pi\Omega DJ_{\Sigma\,(g,A,\pi,\varepsilon)}(h,b,p,f)\label{DHADM2}\\
&+ 16\pi (V_\Sigma \cdot DQ_{\Sigma\, (\varepsilon)}(f)-V_\infty \cdot DQ_{\infty\, (\varepsilon)}(f))\nonumber,
\end{align}
where $\Area_\Sigma(g)$ is the surface area $\Sigma$ and $m(g,\pi)=\sqrt{-\mathbb{P}^\mu\mathbb{P}_\mu}$ is the total mass. Note that the fact $\mathbb{P}^\mu$ is timelike follows from the positive mass theorem, assuming the dominant energy condition (see Theorem 11.2 of \cite{diractype}). Compare this to the first law of black hole mechanics, which states that for perturbations to a stationary solution the following variational formula holds:
\begin{equation}
\delta m=\frac{\kappa}{8\pi}\delta \Area_\Sigma + \Omega\delta J + V\cdot\delta Q.
\end{equation}
This motivates an interesting result of Ashtekar, Fairhurst and Krishnan \cite{AFKfirstlaw} in the framework of isolated horizons. They considered the ADM Hamiltonian on a manifold with an interior boundary representing an isolated horizon, and demonstrated that the validity of the first law is a necessary and sufficient condition for the evolution to be Hamiltonian. However, we take a different approach regarding these additional terms corresponding to the first law. A new Hamiltonian is introduced, \`a la Regge and Teitelboim, that gives the correct equations of motion somewhat more generally, and the first law plays quite a different role. Define the modified Hamiltonian,
\begin{equation}
\mathcal{H}^{RT\,(\xi)}(g,A,\pi,\varepsilon):=16\pi(\xi_\infty\cdot\mathbb{P}+\tilde{J}_\xi-\xi_\Sigma^aQ_{\Sigma \, a})-\int_{\mathcal{M}_0}\xi\cdot\Phi\label{RTHam},
\end{equation}
for some $\xi\in W^{2,2}_{\xiref}$. As before, we fix $\xiref$ on $\Sigma$ such that $\xiref^0=0$, $\xiref^a$ is constant and $\xiref^i$ is tangent to $\Sigma$. Note that $(\xi_\infty\cdot\mathbb{P}+\tilde{J}_\xi-\xi_\Sigma^aQ_{\Sigma \, a})$ only depends on the boundary and asymptotic values of $\xi$, so the Hamiltonian essentially acts as a Lagrange function; extremising the Hamiltonian is equivalent to extremising $(\xi_\infty\cdot\mathbb{P}+\tilde{J}_\xi-\xi_\Sigma^aQ_{\Sigma \, a})$ subject to the constraints being satisfied, where $\xi$ with prescribed boundary and asymptotic conditions, acts as the Lagrange multiplier. This is the basic idea behind Theorem \ref{main2}, below.

Note that the first and last terms in (\ref{RTHam}) are divergent in general, however following Bartnik \cite{phasespace} (see also, Paper I), we combine the integrals and the dominant terms of each cancel out. This leads us to the regularised Hamiltonian,
\begin{align}\label{reghamil}
\mathcal{H}^\xi(g,A,\pi,\varepsilon):=&\int_{\mathcal{M}_0}(\xiref^\alpha-\xi^\alpha)\Phi_\alpha+\int_{\mathcal{M}_0}\xiref^0(\og^{ki}\og^{lj}\onabla_k\onabla_l g_{ij}-\mathring\Delta(\tr_{\mathring{g}}g)\sqrt{\mathring{g}}-\Phi_0)\nonumber\\
&+\int_{\mathcal{M}_0}\og^{ik}\og^{lj}\onabla_k(\xiref^0)(\onabla_j g_{ij}-\onabla_i\tr_{\mathring{g}}g)\sqrt{\mathring{g}}\nonumber\\
&+\int_{\mathcal{M}_0}\xiref^i(\onabla_j(2\pi^j_i+\varepsilon^j_a A^a_i)-\Phi_i)+\int_{\mathcal{M}_0}(2\pi^{i}_j+\varepsilon^i_a A^a_j)\onabla_i\xi^j_{\tref}\nonumber\\
&-\int_{\mathcal{M}_0}\xiref^a(\onabla_i \varepsilon^i_a-\Phi_a)-\int_{\mathcal{M}_0}\varepsilon^i_a\onabla_i\xiref^a,
\end{align}
which is defined on all of $\mathcal{F}$.

\begin{prop}
The regularised Hamiltonian, $\mathcal{H}_\xi:\mathcal{F}\rightarrow\mathbb{R}$, is well-defined and smooth.
\end{prop}
\begin{proof}
This Hamiltonian is exactly of the form considered in Paper I, except that the integrals are performed over a manifold with boundary here, and we have the additional momentum terms, $\int_{\mathcal{M}_0}\xiref^i\onabla_j(\varepsilon^j_aA^a_i)$ and $\int_{\mathcal{M}_0}\varepsilon^i_aA^a_j\onabla_i\xiref^j$. As above, the fact that the manifold has a boundary does not affect the proof at all and, up to the addition of the additional momentum terms, we conclude $\mathcal{H}^\xi$ is smooth from Theorem 4.4 of Paper I. The additional momentum terms are linear in their arguments so simply must be shown to be bounded to demonstrate that they too are smooth. The latter momentum term is clearly bound since $\onabla\xiref$ has bounded support and the former is estimated by (\ref{momentumterminham}).

\end{proof}
An immediate corollary of Theorem 4.2 from Paper I is the following.
\begin{prop}\label{thmadmhamil2}
For $\xi\in W^{2,2}_{-1/2\, c}$, we have
\begin{equation}
\int_{\mathcal{M}_0}\xi\cdot D\Phi_{(\gpi)}(h,b,p,f)=\int_{\mathcal{M}_0}(h,b,p,f)\cdot D\Phi^*_{(\gpi)}(\xi),
\end{equation}
for all $(h,b,p,f)\in T_{(\gpi)}\mathcal{F}$.
\end{prop}
\begin{proof}
The difference, $(h,b,p,f)\cdot D\Phi^*_{(\gpi)}(\xi)-\xi\cdot D\Phi_{(\gpi)}(h,b,p,f)$, is easily computed to give
\begin{align*}
&\nabla^i\Big{(}(\xi^0(\onabla_i\text{tr}_gh-\nabla^jh_{ij})+\onabla^j(\xi^0)h_{ij}-\text{tr}_g h\onabla_i(\xi^0))\sqrt{g}-2\xi^jp_{ij}+\xi^af_{ai}\Big{)}\\
&-\nabla^i\Big{(}2\pi^k_i h_{jk}\xi^j - \pi^{jk}h_{jk}\xi_i +\epsilon_{ijk}b^{ak}B^j_a \xi^0\sqrt{g}+\varepsilon_{ia}b^a_j\xi^j-\xi_i\varepsilon^j_ab^a_j+f_{ai}\xi^jA_j^a\Big{)}.
\end{align*}
The integral of this divergence is then expressed as surface integrals at infinity and on $\Sigma$. The terms at infinity vanish by Theorem 4.2 of Paper I and the terms on $\Sigma$ vanish by the hypothesis $\xi\in W^{2,2}_{-1/2\, c}$. We do have the extra term, $f_{ai}\xi^jA_j^a$, not considered in Paper I, however this clearly vanishes by the same argument.
\end{proof}

\begin{proposition}\label{propreghamil7}
For $\xi\in W^{2,2}_{\xiref}$, the variation of the regularised Hamiltonian is given by
\begin{align}
D\hatH^\xi [h,b,p,f]=&-\oint_\Sigma(\onabla^j(\xi^0)h_{ij}-\text{tr}_g h\onabla_i(\xi^0))\sqrt{g}dS^i-\int_{\mathcal{M}_0}D\Phi^*(\xi)\cdot(h,b,p,f)\label{reghamil3}.
\end{align}
\end{proposition}
\begin{proof}
We consider the terms in (\ref{reghamil}) separately. By Proposition \ref{thmadmhamil2}, the variation of the first integral in (\ref{reghamil}) becomes
\begin{equation*}
\int_\Mo (h,b,p,f)\cdot D\Phi^*(\xiref-\xi).
\end{equation*}
The variation of the second and third terms combine to give 
\begin{align}
\int_\Mo&\Big{\{}\og^{ik}\onabla_k(\xiref^0 \og^{jl}(\onabla_l h_{ij}-\onabla_i h_{jl}))\sqrt{\mathring{g}}-\nabla^i(\xiref^0(\nabla^jh_{ij}-\nabla_i\tr_g h))\sqrt{g}\nonumber\\
&+\nabla^i(h_{ij}\nabla^j\xiref^0-\tr_g h\nabla_i\xiref^0)\sqrt{g}-(h,b,p,f)\cdot D\Phi_{0}^*(\xiref^0)\Big{\}}.\label{blahgh}
\end{align}
Then the first two terms in the above combine to give a total divergence,
\begin{align}
-\oint_{\Mo}&\onabla_k\Big{(}g^{ik}\xiref^0 g^{jl}(\nabla_lh_{ij}-\nabla_i h_{jl})(\sqrt{g}-\sqrt{\og})+(g^{ik}-\og^{ik})\xiref^0 g^{jl}(\nabla_lh_{ij}-\nabla_i  h_{jl})\sqrt{\og}\nonumber\\
&+\og^{ik} \og^{jl}\xiref^0((\nabla_l-\onabla_l)h_{ij}-(\nabla_i-\onabla_i) h_{jl})\sqrt{\og}\Big{)}\label{inftyboundary1},
\end{align}
which is rewritten as surface integrals, both at infinity and on $\Sigma$. The integral at infinity is identical to that considered by Bartnik \cite{phasespace} and therefore vanishes by the same argument, while the surface integral on $\Sigma$ vanishes since $\xi_\Sigma^0=0$. The third term in (\ref{blahgh}) is again a divergence, but only gives a boundary term on $\Sigma$ since $\onabla\xiref$ has bounded support. This boundary term on $\Sigma$ is then exactly the surface integral in (\ref{reghamil3}).

The variation of the fourth and fifth terms in (\ref{reghamil}) give
\begin{align}
\int_\Mo\Big{\{}&2\onabla_i(\xiref^j p^i_j)+2\onabla_j(\xiref^i\pi^{jk}h_{ki})+\nabla_i(\varepsilon^i_ab^a_j\xiref^j)+\nabla_i(f_a^i\xiref^jA_j^a)\nonumber\\
&-2\nabla_i(\xiref^j p^i_j)-2\nabla_i(\pi^{ki}h_{jk}\xiref^j)-\nabla_i(\varepsilon^i_ab^a_j\xiref^j)-\nabla_i(f_a^i\xiref^jA_j^a)\label{cancel32}\\
&+\nabla_i(\xiref^i\varepsilon^j_ab^a_j)-(h,b,p,f)\cdot D\Phi_{i}^*(\xiref^i)\Big{\}}\nonumber,
\end{align}
Since $p$, $\pi$, $f$ and $\varepsilon$ are densities, the divergences above do not depend on the connection used and thus the first two lines in (\ref{cancel32}) cancel exactly. The surface integral on $\Sigma$ arrising from the remaining divergence in \ref{cancel32} vanishes since $\xiref^i$ is tangent to $\Sigma$ and the surface integral at infinity is shown to vanish as follows. Let $S_R=\{x\in M_0\st r(x)=R\}$ and, noting $b$ and $\xiref$ are continuous by the Sobolev-Morrey embedding, we have
\begin{align}
|\oint_{S_R}\xiref^i\varepsilon^j_ab^a_jdS_i|&\lesssim \|b\|_{\infty(S_R)}\|\xiref\|_{\infty(S_R)} \oint_{S_R} |\varepsilon|dS \nonumber \\
&\lesssim o(r^{-1/2})O(1)R^{1/2}\|\varepsilon\|_{1,2,-3/2}\label{vanishYM}\nonumber \\
&=o(1),
\end{align}
where we have made use of the estimate,
\begin{equation*}
\oint_{S_R}|u|dS\leq cR^{1/2}\|u\|_{1,2,-3/2},
\end{equation*}
from \cite{phasespace} (Theorem 4.4). It follows that $\oint_\infty\xiref^i\varepsilon^j_ab^a_jdS_i=0$ and therefore the variation of the fourth and fifth terms in (\ref{reghamil}) reduce to
\begin{equation}
-\int_\Mo (h,b,p,f)\cdot D\Phi_i^*(\xiref^i).
\end{equation}
Finally, the variation of the sixth and seventh terms in (\ref{reghamil}) are given by
\begin{equation}\label{cancel4}
\int_\Mo-\onabla_i(\hat{\xi}^a_\infty f^i_a)+\nabla_i(\hat{\xi}_\infty^af_a^i)-(h,b,p,f)\cdot D\Phi_a^*(\hat{\xi}^a_\infty).
\end{equation}
Since $f$ is a density, the divergences again do not depend on the connection and therefore the first two terms in (\ref{cancel4}) cancel exactly, leaving
\begin{equation}\label{cancel8}
-\int_\mathcal{M}(h,b,p,f)\cdot D\Phi_a^*(\hat{\xi}^a_\infty).
\end{equation}
Assembling all of the pieces completes the proof.
\end{proof}

If $\Sigma$ is indeed the bifurcation surface of a bifurcate Killing horizon, corresponding to the Killing vector $\xi+\xiref$, then $\xi^0=0$ on $\Sigma$ and the surface gravity, $\kappa=\frac{1}{2}n^i\onabla_i(\xi^0)$, is constant. It follows that $\onabla(\xi^0)$ is normal to $\Sigma$ and making use of coordinates adapted to $\Sigma$, the surface integral in (\ref{reghamil3}) becomes
\begin{align}
-\oint_\Sigma&\left(\onabla^j(\xi^0)h_{ij}-\text{tr}_g h\onabla_i(\xi^0)\right)\sqrt{g}dS^i\nonumber\\
&=-\oint_\Sigma\left(g^{j3}\onabla_3(\xi^0)h_{ij}n^i-h_k^k\onabla_3(\xi^0)\right)\sqrt{g}dS\nonumber\\
&=\oint_\Sigma\onabla_3(\xi^0)h_A^A\sqrt{g}dS\nonumber\\
&=2\kappa d\Area_\Sigma\label{dAformula},
\end{align}
where the index `3' refers to the direction normal to $\Sigma$, while $A=1,2$ are tangential.

It can be seen from Proposition \ref{propreghamil7}, that this new Hamiltonian gives the correct equations of motion when $\onabla\xi^0\equiv0$ on $\Sigma$, or when $\Sigma$ is the bifurcation surface of a bifurcate Killing Horizon and $g$ is a critical point of the area functional of $\Sigma$. Since the surface gravity explicitly depends on $g$, there is no obvious way to further modify the Hamiltonian such that the correct equations of motion are generated in general.

To prove the main Theorem, we will need to make use of the following generalisation of the method of Lagrange multipliers to Banach manifolds (see, Theorem 6.3 of \cite{phasespace}).
\begin{theorem}\label{banach}
Suppose $K:B_1\rightarrow B_2$ is a $C^1$ map between Banach manifolds, such that $DK_u:T_uB_1\rightarrow T_{K(u)}B_2$ is surjective, with closed kernel and closed complementary subspace for all $u\in K^{-1}(0)$. Let $f\in C^1(B_1)$ and fix $u\in K^{-1}(0)$, then the following statements are equivalent:
\begin{enumerate}[(i)]
\item For all $v\in\ker DK_u$, we have
\begin{equation}Df_u(v)=0.\end{equation}
\item There is $\lambda\in B_2^*$ such that for all $v\in B_1$,
\begin{equation}Df_u(v)=\left<\lambda,DK_u(v)\right>,\end{equation}
where $\left< \, , \right>$ refers to the natural dual pairing.
\end{enumerate}
\end{theorem}

We also will need to make use of the following Theorem from Paper I, regarding weak solutions. A weak solution of $D\Phi^*(\xi)=f$ is an element $\xi\in \mathcal{N}$ such that
\begin{equation}\label{weaksoln}
\int_\mathcal{M}\xi\cdot D\Phi(h,b,p,f)=\int_\mathcal{M} f\cdot (h,b,p,f),
\end{equation}
for all $(h,b,p,f)\in\mathcal{G}\times\mathcal{A}\times\mathcal{K}\times\mathcal{E}=T_{(\gpi)}\mathcal{F}$.
\begin{theorem}\label{thmweakstrong}
If $\xi\in \mathcal{N}$ is a weak solution of $D\Phi_{(\gpi)}^*(\xi)=(f_1,f_2,f_3,f_4)$, with $(f_1,f_3,f_4)\in L^2_{-5/2}\times W^{1,2}_{-3/2}\times W^{1,2}_{-3/2}$ and $(\gpi)\in\mathcal{F}$, then $\xi\in W^{2,2}_{-1/2}$ and is indeed a strong solution.
\end{theorem}
The following theorem from $\cite{phasespace}$ is stated in reference to a particular operator, however it is clear from the proof that the theorem applies to a general class of operators. In particular, the theorem could more generally be stated as follows:
\begin{theorem}[Theorem 3.6 of \cite{phasespace}]\label{thmtrivker}
Let $\Omega\subset\mathcal{M}$ be a connected domain with $E'_R\subset\Omega$ for some $R$, where $E'_R$ is a connected component of $E_R$. If $\xi\in W^{2,2}_{-1/2}$ satisfies
\begin{equation*}
\onabla^2\xi=b_1\nabla\xi+b_0\xi,
\end{equation*}
with $b_0\in L^2_{-5/2}$ and $b_1\in W^{1,2}_{-3/2}$, then $\xi\equiv 0$ in $\Omega$.
\end{theorem}
From this and Theorem \ref{thmweakstrong}, we have the following immediate corollary:
\begin{corollary}\label{cortrivker}
Let $(\gpi)\in\mathcal{F}$. If $\xi\in \mathcal{N}^*$ satisfies $D\Phi_{(\gpi)}^*(\xi)=0$ on a connected $\Omega\subset\mathcal{M}$ containing some $E_R'$, then $\xi\equiv0$ on $\Omega$.
\end{corollary}

Now we are in a position to prove the main theorem. Below, we use the notation $D\Phi^*_{(\gpi)}(\xi)=(D\Phi^*_g(\xi),D\Phi^*_A(\xi),D\Phi^*_\pi(\xi),D\Phi^*_\varepsilon(\xi))$ to identify the components of $D\Phi^*$.
\begin{theorem}\label{main2}
Let $(g,A,\pi,\varepsilon)\in\mathcal{C}(s)$, where $s\in L^1$, and suppose there exists a vector field, $\phi\in W^{2,2}_{\text{loc}}$, tangent to $\Sigma$ with $D\Phi^*_\pi(\phi),D\Phi^*_\varepsilon(\phi)\in W^{1,2}_{-1/2\, c}(\mathcal{M}_0)$. Further suppose that for all $(h,b,p,f)\in T_{(\gpi)}\mathcal{C}(s)$,
\begin{align}
Dm_{(\gpi)}&(h,b,p,f)=\alpha D\Area_{\Sigma\,(\gpi)}(h,b,p,f)+\beta DJ_{\phi \, (\gpi)}(h,b,p,f)\nonumber\\
&+\gamma_\Sigma\cdot DQ_{\Sigma \, (\gpi)}(h,b,p,f)-\gamma_\infty\cdot DQ_{\infty \, (\gpi)}(h,b,p,f)\label{firstlawthm},
\end{align}
where $\alpha,\beta\in\mathbb{R}$ and $\gamma_{\Sigma},\gamma_{\infty}\in\lie{g}$ are constants. Then $(\gpi)$ is a generalised stationary initial data set. Furthermore, $\gamma$ is the electric potential, and if $\Sigma$ is the bifurcation surface of a bifurcate Killing horizon, then $8\pi\alpha$ is the surface gravity and $\beta$ is the angular velocity.
\end{theorem}
\begin{proof}
Assume (\ref{firstlawthm}) holds at some fixed point $\tilde{G}=(\tilde{g},\tilde{A},\tilde{\pi},\tilde{\varepsilon})\in\mathcal{F}$. Then fix $\xiref$ such that it satisfies the following boundary and asymptotic conditions:
\begin{itemize}
\item $\xi^\mu_{\infty}$ corresponds to a future pointing unit vector at spatial infinity in the spacetime that is proportional to $\mathbb{P}^\mu$,
\item $\xiref^a$ is constant at infinity and on $\Sigma$, with values $\xi_\infty^a=\gamma^a_\infty$ and $\xi_\Sigma^a=\gamma_\infty^a$,
\item $\xiref^0$ vanishes on $\Sigma$,
\item $\xiref^i=-\beta\phi^i$ on $\Sigma$,
\item $\partial_i(\xiref^0)\tilde{n}^i=16\pi\alpha$ on $\Sigma$.
\end{itemize}
We use $\tilde{n}$ to denote the unit normal with respect to $\tilde{g}$, pointing towards infinity in $M_0$. Note that the condition on $\xi^\mu_\infty$ implies $\xi^\mu_{\infty}\mathbb{P}_\mu=m$, and the conditions on $\alpha$, $\beta$ and $\gamma$ ensure that they correspond to the appropriate physical quantities in the statement of the Theorem.

Now for some $\xi\in W^{2,2}_{\xiref}$, define
\begin{equation}
\tilde{f}(G):=\mathcal{H}^\xi(G)-16\pi\alpha \Area_\Sigma(G),
\end{equation}
where $G=(\gpi)\in\mathcal{F}$. We again let $K(G)=\Phi(G)-s$, and note that for all constrained variations, $(h,b,p,f)\in \ker(DK_{\tilde{G}})=T_{\tilde{G}}\mathcal{C}(s)$, we have (see \ref{RTHam})
\begin{align*}
D\mathcal{H}^\xi_{\tilde{G}}(h,b,p,f)=&16\pi(\xi_\infty\cdot D\mathbb{P}_{\tilde{G}}(h,b,p,f)+D\tilde{J}^{\xi}_{\tilde{G}}(h,b,p,f)-\xi_\Sigma^a DQ_{\Sigma \,{\tilde{G}}\,a}(h,b,p,f))\\
=&Dm_{{\tilde{G}}}(h,b,p,f)-\beta DJ_{\phi \, {\tilde{G}}}(h,b,p,f)\nonumber\\
&-\gamma_\Sigma\cdot DQ_{\Sigma \, {\tilde{G}}}(h,b,p,f)+\gamma_\infty\cdot DQ_{\infty \, {\tilde{G}}}(h,b,p,f).
\end{align*}
By hypothesis (\ref{firstlawthm}), we have $D\tilde{f}_{\tilde{G}}(h,b,p,f)=0$ for all $(h,b,p,f)\in \ker(DK_{\tilde{G}})$. It follows from Theorem \ref{banach}, that there exists $\lambda\in\mathcal{N}$ such that
\begin{equation}
D\tilde{f}_{\tilde{G}}=\left<D\Phi_{\tilde{G}},\lambda\right>;
\end{equation}
that is,
\begin{equation}
D\tilde{f}_{\tilde{G}}(h,b,p,f)=\int_\mathcal{M}D\Phi_{\tilde{G}}(h,b,p,f)\cdot\lambda,
\end{equation}
for all $(h,b,p,f)\in T_{\tilde{G}}\mathcal{F}$. However, from Proposition \ref{propreghamil7}, we have
\begin{align}
D\tilde{f}_{\tilde{G}}(h,b,p,f)=&-\oint_\Sigma(\onabla^j(\xi^0)h_{ij}-\text{tr}_g h\onabla_i(\xi^0))\sqrt{g}dS^i\\
&-\int_{\mathcal{M}_0}D\Phi^*(\xi)\cdot(h,b,p,f)-16\pi\alpha D\Area_{\Sigma\, (\tilde{G})}(h,b,p,f)\nonumber.
\end{align}
As $\partial_i(\xi^0)\tilde{n}^i=16\pi\alpha$ on $\Sigma$, the first and last terms cancel exactly (see (\ref{dAformula})), leaving
\begin{equation}
D\tilde{f}_{\tilde{G}}(h,b,p,f)=-\int_{\Mo}(h,b,p,f)\cdot D\Phi^*_G(\xi);
\end{equation}
that is,
\begin{equation}
-\int_{\Mo}(h,b,p,f)\cdot D\Phi^*_G(\xi)=\int_\mathcal{M}D\Phi_{\tilde{G}}(h,b,p,f)\cdot\lambda\label{bhghghg},
\end{equation}
for all $(h,b,p,f)\in T_{(\tilde{G})}\mathcal{F}$.

Since the first integral in (\ref{bhghghg}) is over $\Mo$, rather than $\mathcal{M}$, Theorem \ref{thmweakstrong} does not directly apply. Instead we extend $D\Phi_{\tilde{G}}^*(\xi)$ by zero, noting that the hypotheses on $D\Phi_{\tilde{G}}^*(\phi)$ ensure that we can do this without losing regularity.

Define the function
\begin{equation}
\psi=(\psi_1,\psi_2,\psi_3,\psi_4):=\left\{\begin{matrix}-D\Phi_{\tilde{G}}^*(\xi) & \text{on }\mathcal{M}_0\\ 0 & \text{otherwise}\end{matrix}\right. .
\end{equation}
We then have
\begin{equation}
\int_{\mathcal{M}}\psi\cdot(h,b,p,f)=\int_\mathcal{M}D\Phi_{\tilde{G}}(h,b,p,f)\cdot\lambda
\end{equation}
for all $(h,b,p,f)\in T_{\tilde{G}}\mathcal{F}$. It is straightforward to check that $\psi_1\in L^2_{-5/2}(\mathcal{M})$ and $\psi_3,\psi_4\in W^{1,2}_{-3/2}(\mathcal{M})$ (see Lemma 6.5 of \cite{MyThesis} for details), and therefore Theorem \ref{thmweakstrong} gives $\lambda\in W^{2,2}_{-1/2}(\mathcal{M})$ and $D\Phi_{\tilde{G}}^*(\lambda)=\psi$ in the strong sense. It then follows that $D\Phi^*_{\tilde{G}}(\tilde{\xi})=0$ on $\mathcal{M}_0$, where $\tilde{\xi}:=\xi+\lambda$ is the generalised stationary Killing vector.

\end{proof}

Note that we have $D\Phi^*_{\tilde{G}}(\lambda)=0$ on $\mathcal{M}\setminus\Mo$, so Corollary \ref{cortrivker} implies $\lambda=0$ on $\mathcal{M}\setminus\Mo$. It then follows that $\tilde{\xi}=\xi=-\beta\phi$ on $\Sigma$, and in particular we have that $\tilde{\xi}+\beta\phi^i$ vanishes on $\Sigma$. It is interesting to note that while we do not assume that $\Sigma$ is a horizon in the above theorem, the conclusion that $\tilde{\xi}^\mu+\beta\phi^i$ vanishes on $\Sigma$ gives us the following corollary.
\begin{corollary}
If the hypotheses of Theorem \ref{main2} hold and $(\gpi)$ is axially symmetric with axial Killing field, $\phi$, then $\Sigma$ is the bifurcation surface of a bifurcate Killing horizon, where $8\pi\alpha$ is the surface gravity and $\beta$ is the angular velocity.
\end{corollary}
\begin{proof}
This is an immediate consequence of the fact that if a Killing field vanishes on a spacelike $2$-surface then that surface is the bifurcation surface of a bifurcate Killing horizon (see, for example, Chapter 5 of \cite{waldQFT}).
\end{proof}
\begin{remark}
By virtue of the fact that $D\Phi^*(\xi)=0$ for a Killing vector, $\xi$, we do indeed have $D\Phi^*_\pi(\phi),D\Phi^*_\varepsilon(\phi)\in W^{1,2}_{-1/2\, c}(\mathcal{M}_0)$ when $\phi$ is the axial Killing vector.
\end{remark}

\section{Acknowledgements}
This article consists of research conducted during my candidature as a PhD student at Monash University. I would like to express my gratitude for the support provided by Monash University and my supervisors, Robert Bartnik and Todd Oliynyk, during this time.

\bibliographystyle{plain}
\bibliography{../../refs}

\begin{thebibliography}{10}

\bibitem{ArmsEYM}
J.~Arms.
\newblock Linearization stability of gravitational and gauge fields.
\newblock {\em J. Math. Phys.}, 20(3):443--453, 1979.

\bibitem{etal1}
J.~Arms, J.~E. Marsden, and V.~Moncrief.
\newblock The structure of the space of solutions of {E}instein's equations.
  {II}. {S}everal {K}illing fields and the {E}instein-{Y}ang-{M}ills equations.
\newblock {\em Ann. Phys.}, 144(1):81--106, 1982.

\bibitem{AFKfirstlaw}
A.~Ashtekar, S.~Fairhurst, and B.~Krishnan.
\newblock Isolated horizons: {H}amiltonian evolution and the first law.
\newblock {\em Phys. Rev. D}, 62(10):104025, 2000.

\bibitem{AF}
R.~Bartnik.
\newblock The mass of an asymptotically flat manifold.
\newblock {\em Comm. Pure. Appl. Math.}, 19:661--693, 1986.

\bibitem{phasespace}
R.~Bartnik.
\newblock Phase space for the {E}instein equations.
\newblock {\em Comm. Anal. Geom.}, 13(5):845--885, 2005.

\bibitem{diractype}
R.~Bartnik and P.~Chru{\'s}ciel.
\newblock Boundary value problems for {D}irac-type equations.
\newblock {\em Journal f{\"u}r die reine und angewandte Mathematik},
  2005(579):13--73, 2005.

\bibitem{BDF}
D.~Brill, S.~Deser, and L.~Fadeev.
\newblock Sign of gravitational energy.
\newblock {\em Phys. Lett. A}, 26(11):538--539, 1968.

\bibitem{ellipticsys}
Y.~Choquet-Bruhat and D.~Christodoulou.
\newblock Elliptic systems in spaces on {$H_{s,\delta}$} manifolds which are
  euclidean at infinity.
\newblock {\em Acta Mathematica}, 146(1):129--150, 1981.

\bibitem{globalcharges}
P.~Chru\'sciel and W.~Kondracki.
\newblock Some global charges in classical {Y}ang-{M}ills theory.
\newblock {\em Phys. Rev. D}, 36(6):1874--1881, 1987.

\bibitem{ghanem}
S.~Ghanem.
\newblock The global existence of {Y}ang-{M}ills fields on curved space-times.
\newblock {\em arXiv preprint arXiv:1312.5476}, 2013.

\bibitem{boundedl2}
S.~Klainerman, I.~Rodnianski, and J.~Szeftel.
\newblock The bounded {$L^2$} curvature conjecture.
\newblock {\em arXiv preprint arXiv:1204.1767}, 2012.

\bibitem{MyThesis}
S.~McCormick.
\newblock {\em The phase space for the {E}instein-{Y}ang-{M}ills equations,
  black hole mechanics, and a condition for stationarity}.
\newblock PhD thesis, Monash University, 2014.

\bibitem{Me}
S.~McCormick.
\newblock The phase space for the {E}instein-{Y}ang-{M}ills equations and the
  first law of black hole thermodynamics.
\newblock {\em Advances in Theoretical and Mathematical Physics}, 18(4), 2014,
  (to appear).

\bibitem{Moncrief1}
V.~Moncrief.
\newblock Spacetime symmetries and linearization stability of the {E}instein
  equations. {I}.
\newblock {\em J. Math. Phys.}, 16(3):493--498, 1975.

\bibitem{Moncrief2}
V.~Moncrief.
\newblock Space-time symmetries and linearization stability of the {E}instein
  equations. {II}.
\newblock {\em J. Math. Phys.}, 17(10):1893--1902, 1976.

\bibitem{NW}
L.~Nirenberg and H.~F. Walker.
\newblock The null spaces of elliptic partial differential operators in
  {$\mathbb{R}^n$}.
\newblock {\em Journal of Mathematical Analysis and Applications},
  42(2):271--301, 1973.

\bibitem{RT}
T.~Regge and C.~Teitelboim.
\newblock The role of surface integrals in general relativity.
\newblock {\em Ann. Phys.}, (88):286--318, 1974.

\bibitem{SW1}
D.~Sudarsky and R.~M. Wald.
\newblock Extrema of mass, stationarity, and staticity, and solutions to the
  {E}instein-{Y}ang-{M}ills equations.
\newblock {\em Phys. Rev. D}, 46(4):1453--1474, 1992.

\bibitem{QLMreview}
L.~B. Szabados.
\newblock Quasi-local energy-momentum and angular momentum in general
  relativity: {A} review article.
\newblock {\em Living Rev. Relativity}, 7(4), 2004.

\bibitem{waldQFT}
R.~M. Wald.
\newblock {\em Quantum field theory in curved spacetime and black hole
  thermodynamics}.
\newblock University of Chicago Press, 1994.

\bibitem{Wald}
R.~M. Wald.
\newblock {\em General relativity}.
\newblock University of Chicago press, 2010.

\end{thebibliography}
\end{document}